\documentclass[a4paper,UKenglish, thm-restate, cleveref, autoref, notab, nolineno]{socg-lipics-v2021}

\usepackage{todonotes}

\usepackage{amssymb}
\usepackage{amsmath}
\usepackage{xurl} 
\usepackage{xspace}

\usepackage{array}
\usepackage{booktabs}
\usepackage{multirow}
\newcolumntype{C}[1]{>{\centering\arraybackslash}p{#1}}
\newcolumntype{R}[1]{>{\raggedleft\arraybackslash}p{#1}}

\newcommand{\pl}{PL\xspace}
\newcommand{\real}[1]{\mathbb{R}^{#1}}
\newcommand{\revision}[1]{{#1}}

\title{Singular Arrange and Traverse Algorithm for Computing Reeb Spaces of Bivariate PL Maps} 

\titlerunning{Singular Arrange and Traverse Algorithm for Reeb Spaces} 

\author{Petar Hristov\footnote{Corresponding author}}{Scientific Visualization Group, Department of Science and Technology (ITN), Linköping University, Sweden}{petar.hristov@liu.se}{[https://orcid.org/0000-0002-1482-2529]}{}

\author{Ingrid Hotz}{Scientific Visualization Group, Department of Science and Technology (ITN), Linköping University, Sweden}{ingrid.hotz@liu.se}{[https://orcid.org/0000-0001-7285-0483]}{The author acknowledges support from the Swedish e-Science Research Center (SeRC), the ELLIIT environment for strategic research in Sweden, and the Swedish Research Council~(VR) grant 2019-05487.}

\author{Talha Bin Masood}{Scientific Visualization Group, Department of Science and Technology (ITN), Linköping University, Sweden}{talha.bin.masood@liu.se}{[https://orcid.org/0000-0001-5352-1086]}{This work was primarily supported by the WASP-DDLS grant from Wallenberg Autonomous Systems and Software Program (WASP) funded by the Knut and Alice Wallenberg Foundation. The author acknowledges Swedish Research Council (VR) grant 2023-04806 and the Swedish e-Science Research Center (SeRC) for additional funding support.}

\authorrunning{P. Hristov, I. Hotz, and T. Bin Masood} 

\Copyright{Petar Hristov, Ingrid Hotz, and Talha Bin Masood} 



\category{} 

\relatedversion{} 

\acknowledgements{We thank Dr. Nanna H. List (KTH Royal Institute of Technology, Stockholm) for providing the MVK dataset used for evaluation in this paper. The Ethane-diol dataset is taken from the TTK dataset repository (\url{https://github.com/topology-tool-kit/ttk-data}) with acknowledgments to Roberto Alvarez Boto.
The enzo dataset is taken from the \emph{The Enzo Collaboration} with acknowledgments to Michael Norman (University of California, San Diego, USA).}

\nolinenumbers 

\supplementdetails{(Software Implementation)}{https://doi.org/10.5281/zenodo.18759995}

\begin{document}

\maketitle

\begin{abstract}
We present an exact and efficient algorithm for computing the Reeb space of a bivariate \pl map.
The Reeb space is a topological structure that generalizes the Reeb graph to the setting of multiple scalar-valued functions defined over a shared domain, a situation that frequently arises in practical applications.
While the Reeb graph has become a standard tool in computer graphics, shape analysis, and scientific visualization, the Reeb space is still in the early stages of adoption.
Although several algorithms for computing the Reeb space have been proposed, none offer an implementation that is both exact and efficient, which has substantially limited its practical use.
To address this gap, we introduce \emph{singular arrange and traverse}, a new algorithm built upon the \emph{arrange and traverse} framework~\cite{hristovArrangeTraverseAlgorithm2025}. 
Our method exploits the fact that, in the bivariate case, only singular edges contribute to the structure of Reeb space, allowing us to ignore many regular edges~\cite{tiernyJacobiFiberSurfaces2017}.
This observation results in substantial efficiency gains on datasets where most edges are regular, which is common in many numerical simulations of physical systems.
We provide an implementation of our method and benchmark it against the original \emph{arrange and traverse} algorithm, showing performance gains of up to four orders of magnitude on real-world datasets.
\end{abstract}

\section{Introduction}
\label{sec:introduction}

Visualizing and analyzing the structure of scalar-valued functions is fundamental across many disciplines, including physics~\cite{sousbiePersistentCosmicWeb2011}, chemistry~\cite{bhatiaTopoMSComprehensiveTopological2018, guntherCharacterizingMolecularInteractions2014}, atmosphere science~\cite{kuhnTopologyBasedAnalysisMultimodal2017, tricocheTopologicalMethodsVisualizing2009}, biology~\cite{ nicolauTopologyBasedData2011, reimannCliquesNeuronsBound2017}, computer graphics~\cite{biasottiReebGraphsShape2008, hachaniSegmentation3DDynamic2014}, and machine learning~\cite{cloughTopologicalLossFunction2022a, naitzatTopologyDeepNeural2020}.
Topological tools such as persistent homology, the Morse-Smale complex, and the Reeb graph are state-of-the-art methods for extracting meaningful features from scalar fields~\cite{heineSurveyTopologybasedMethods2016a}.
Their widespread adoption is largely due to the fact that they can be computed \textit{robustly} and \textit{efficiently}. 


The multivariate setting, where multiple scalar functions are defined over a shared domain, is considerably more complex. In this context, the interrelated structure between the functions cannot be captured adequately by scalar field techniques.  To address this limitation, the Reeb graph can be generalized to the Reeb space~\cite{edelsbrunnerReebSpacesPiecewise2008b}.
The Reeb space is a quotient space that contracts the components of all preimages of points, also called \emph{fibers}.
The Reeb space of a bivariate \pl map over a compact \pl manifold can be represented as a two-dimensional cell complex~\cite{edelsbrunnerReebSpacesPiecewise2008b, hristovHypersweepsConvectiveClouds2022}, where the 2-cells, referred to as sheets, represent features in data.

Four algorithms have been proposed for computing the Reeb space of a \pl map.
The first algorithm~\cite{edelsbrunnerReebSpacesPiecewise2008b} segments the domain into regions where fibers have uniform connectivity and then glues those regions appropriately to form the sheets of the Reeb space.
The Jacobi fiber surfaces algorithm~\cite{tiernyJacobiFiberSurfaces2017} significantly improves practical performance in the bivariate case by computing these regions only for the singular (or Jacobi) edges in the domain, where fibers can change their \revision{connectivity}. 
The third algorithm uses a hierarchical decomposition-based approach~\cite{chattopadhyayAlgorithmFastCorrect2024} to compute the Reeb space of a bivariate \pl map from a sequence of nested Reeb graphs.
Finally, the arrange and traverse algorithm~\cite{hristovArrangeTraverseAlgorithm2025} computes combinatorial descriptions of fibers called preimage graphs and determines their topological correspondence to obtain regions of uniform fiber connectivity that match the sheets of the Reeb space.

While all four Reeb space algorithms have close to quadratic algorithmic complexity in the bivariate case, their implementations suffer from either efficiency or robustness issues.
The arrange and traverse algorithm~\cite{hristovArrangeTraverseAlgorithm2025} provides a robust implementation~\cite{hristov2025} in the bivariate case with CGAL's exact implementation of arrangements of 2D line segments~\cite{fabriDesignCGALComputational2000}. 
However, due to its slow running time, it cannot handle real-life datasets without significant downsampling.
The Jacobi fiber surfaces algorithm~\cite{tiernyJacobiFiberSurfaces2017} offers better performance, but, as has been noted~\cite{chattopadhyayAlgorithmFastCorrect2024}, its implementation is not exact without additional care in handling Jacobi fiber surface intersections.
This paper aims to combine the strengths of these two approaches to provide a Reeb space algorithm with an implementation that is both efficient and robust.

We propose the \textit{singular arrange and traverse} algorithm for computing the Reeb space of a generic bivariate \pl map $f : |K| \to \real{2}$, defined over a triangulation of a 3-manifold $K$.
This is a specialization of the arrange and traverse algorithm that does not require the arrangement of the images of all edges of $K$ in the range, only the singular ones, where the \revision{connectivity} of the fibers can change.
Using red-blue line segment intersection on the images of regular and singular edges, our algorithm computes the Reeb space from a subset of all preimage graphs, which we call \emph{essential}.
The worst-case running time of our algorithm is $O(N_s N_t \log N_t)$, where $N_s$ is the number of singular edges and $N_t$ is the number of triangles.
\revision{
This is a significant improvement because datasets with a low \revision{proportion} of singular edges are common in practice.
We provide an exact implementation and benchmark it against the arrange and traverse algorithm, the only other exact implementation known to the authors, achieving up to four orders of magnitude improvement in running time and memory footprint.}

This paper is structured as follows: in \cref{sec:preliminaries}, we introduce the background relevant to Reeb spaces of \pl maps; in \cref{sec:algorithm}, we present the singular arrange and traverse algorithm; in \cref{sec:proofs}, we prove its correctness and worst-case complexity; in \cref{sec:implementation} we describe our implementation, which we benchmark in \cref{sec:evaluation}. We conclude in \cref{sec:conclusion}.

\section{Preliminaries}
\label{sec:preliminaries}

Let $K$ be a triangulation of a $3$-manifold and let $f: |K| \to \real{2}$ be a \revision{piecewise linear (\pl) map~\cite{rourkeIntroductionPiecewiseLinearTopology1982, edelsbrunnerComputationalTopology2009}; see \cref{fig:domain-regular-arrangement}}.
We will refer to $f$ as a \emph{bivariate} \pl map because it can be decomposed into two scalar \pl maps $f = (f_1, f_2)$, where $f_1, f_2 : |K| \to \real{}$.
We assume that $f$ is \emph{generic}, in the sense that the images of no three vertices are collinear and \revision{the images of no three edges intersect at a single point in their interiors}~\cite{hristovRobustGeometricPredicates2025}.
We refer to the images of edges under $f$ in the range as \emph{segments} and to the image of vertices under $f$ as \emph{vertex points}.
For an edge $e$ with endpoints $a$ and $b$, we will also refer to $e$ as $ab$.

\begin{figure}[ht]
  \centering
  \subfloat[Tetrahedral mesh.]{
    \includegraphics[width=0.30\textwidth]{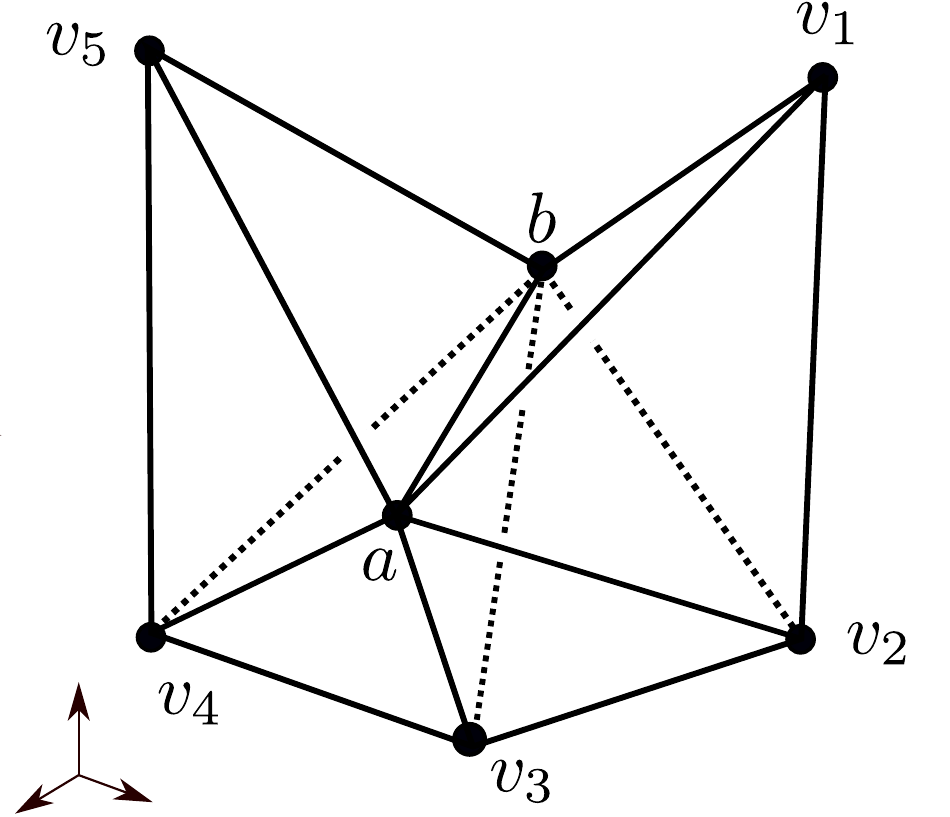}
  }
  \quad
  \captionsetup{textformat=simple}
  \subfloat[Image in the range.]{
    \includegraphics[width=0.30\textwidth]{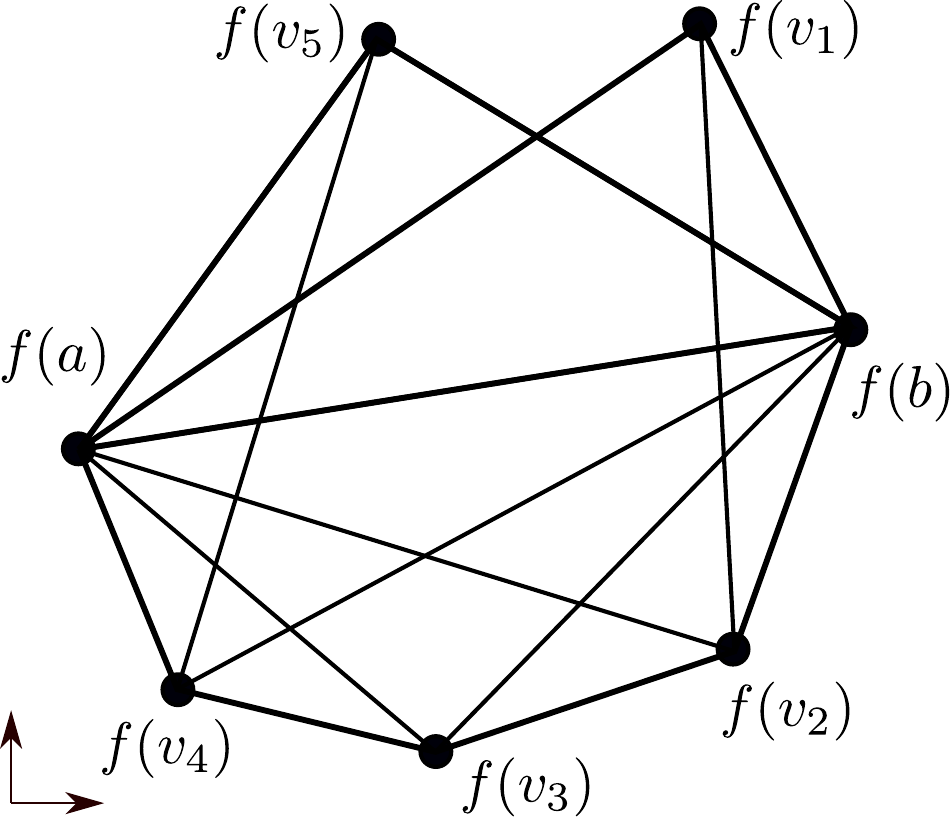}
    \label{fig:domain-regular-arrangement-sub2}
  }
  \caption{Example tetrahedral mesh (a) and its image in the range under a bivariate \pl map (b).}
  \label{fig:domain-regular-arrangement}
\end{figure}

The \emph{upper link} of an edge $ab$ is a sub-complex of the link of $ab$ induced by the vertices $v$ in the link, whose vertex points $f(v)$ have a positive orientation with respect to $f(a)$ and $f(b)$~\cite{hristovHypersweepsConvectiveClouds2022}.
The \emph{lower link} is defined analogously by the vertex points with negative orientation.
\revision{
In general we will assume that the endpoints of an edge or a segment are ordered lexicographically in order to establish a canonical orientation.
For example, the upper link of $ab$ in \cref{fig:domain-regular-arrangement} is $\{v_1, v_5\}$ and the lower link is $\{v_2, v_3, v_4, v_2v_3, v_3v_4\}$.
}

An edge is \emph{regular} if both its lower and upper links have one \revision{simply} connected component.
Otherwise, that edge is \emph{singular}, or Jacobi~\cite{edelsbrunnerJacobiSetsMultiple2002}.
A vertex is \emph{singular} if it is incident to at least one singular edge and \emph{regular} otherwise.
We call the images of singular edges and vertices in the range \emph{singular segments and vertex points} respectively.
The collection of singular edges and vertices is a one-dimensional subcomplex of $K$, which we call the \emph{singular set} $S_f$.

A \emph{fiber} is the preimage of a point $f^{-1}((x_1,x_2)) \subset |K|$ for $(x_1, x_2) \in \real{2}$, which we call a \emph{fiber point} (see \cref{fig:example-fibers-reeb-space}).
A simplex intersected by a fiber is called \emph{active} \revision{with respect to that fiber}.
We refer to the connected components of fibers as \emph{fiber components}.
Since fiber components can change their connectivity only at the singular set~\cite{ hristovHypersweepsConvectiveClouds2022, tiernyJacobiFiberSurfaces2017}, those that intersect $S_f$ are called \emph{singular}, and those that do not are called \emph{regular}.
\revision{Regular fiber components are polylines, while singular fiber components are either an isolated point or a number of polylines glued together at exactly one or two points~\cite{hristov2025}.}
Fibers with only regular components are \emph{regular}, otherwise, they are \emph{singular}.

Singular edges, analogous to critical points in the univariate case, can be categorized in the following way: \emph{definite} edges create or destroy fiber components (analogous to minima and maxima) and \emph{indefinite} edges merge or split fiber components (analogous to saddles)~\cite{hristovHypersweepsConvectiveClouds2022, hristovArrangeTraverseAlgorithm2025}.
Indefinite edges that split or merge exactly two components are referred to as \emph{simple}.
The singular set is called \emph{simple}~\cite{ edelsbrunnerReebSpacesPiecewise2008b, chattopadhyayAlgorithmFastCorrect2024} when it is a collection of \revision{polylines} and all indefinite edges are simple.
Non-simple singular edges can be unfolded into multiple simple ones (like monkey saddles), \revision{but that requires the addition of new simplices to the input mesh}~\cite{edelsbrunnerJacobiSetsMultiple2002}.

The \emph{Reeb space} $R_f$ of a \pl map $f$ is a topological structure that contracts each fiber component to a single point.
It is the quotient space of $|K|$ with respect to the equivalence relation $\sim$, where $x \sim y$ if and only if $x$ and $y$ belong to the same fiber component.
The Reeb space is a part of the Stein factorization $f = \pi \circ q$, where $q : |K| \to R_f$ is the quotient map with respect to $\sim$ and $\pi$ is the natural map such that the diagram commutes~\cite{mottaStableMaps3Manifolds1995}.
For computational purposes, the Reeb space of a bivariate \pl map can be represented as a two-dimensional cell complex whose 2-cells, referred to as \emph{sheets}, are glued together along the image of the singular set under $q$~\cite{chattopadhyayExtractingJacobiStructures2014}.
For an example see \cref{fig:example-fibers-reeb-space}.

\begin{figure}[ht]
  \centering
  \includegraphics[width=1.00\textwidth]{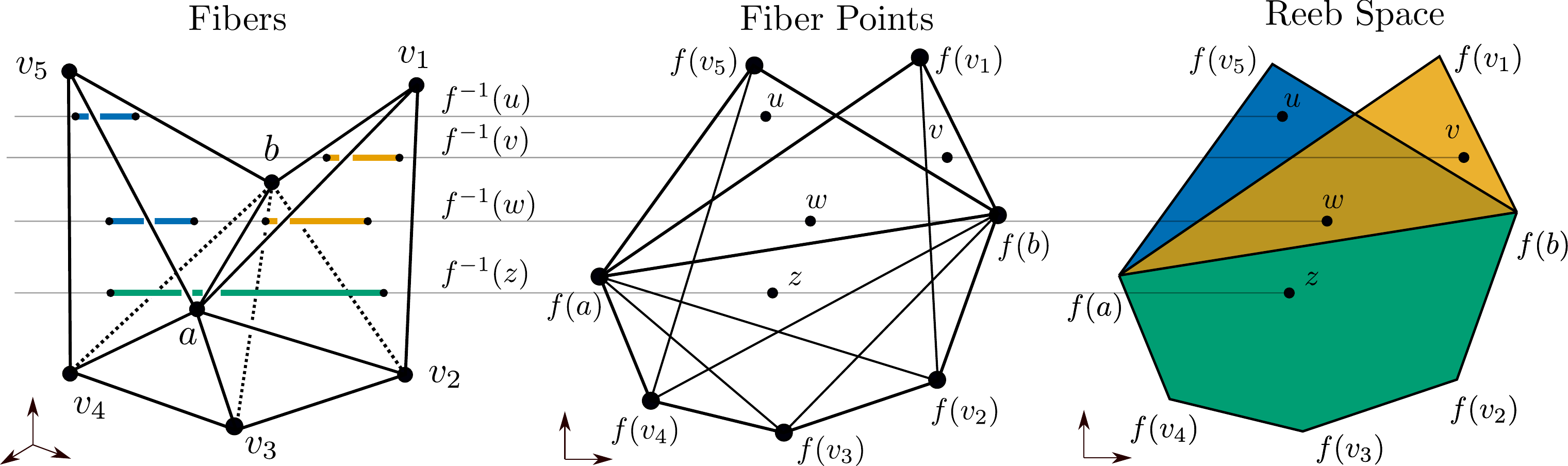}
  \caption{
      Examples of the change of connectivity of four fibers at the indefinite edge $ab$ and the definite edges $av_1$ and $bv_5$.
      The Reeb space of $f$ has three sheets (blue, yellow and green) -- one for each fiber component.
      To visualize the Reeb space we map its sheets to the range and overlay them.
  }
  \label{fig:example-fibers-reeb-space}
\end{figure}

\subsection{Arrange and Traverse Algorithm}
\label{sec:background-arrange-and-traverse}

\revision{The Reeb space of a \pl map can be computed with four different approaches, which we outlined in \cref{sec:introduction}.
We will review the \emph{arrange and traverse} algorithm~\cite{hristovArrangeTraverseAlgorithm2025} in detail because it forms the foundation for our new approach.}
The \emph{arrange and traverse} algorithm computes the Reeb space by combinatorially representing fibers and tracking how their connectivity changes as they intersect the edges of $K$. 
Fibers are represented combinatorially with a preimage graph, similar to Parsa's optimal Reeb graph algorithm~\cite{parsa2012}.
\revision{The \emph{preimage graph} of a fiber $f^{-1}((x_1,x_2)) \subset |K|$ for a point $(x_1, x_2) \in \real{2}$, which is not in the image of an edge or vertex,} is a graph whose vertices are the active triangles with respect to that fiber and whose edges connect triangles which are both the faces of the same tetrahedron.

In the first stage of this algorithm, the edges of $K$ are mapped to the range \revision{via $f$}, and their geometric arrangement $A$ is computed (see \cref{sec:background-intersections}).
Within each face $F$ of $A$, all fibers have the same active triangles, so they can all be represented by a single \emph{preimage graph}.
In the second stage, the faces of $A$ are traversed in order to compute their preimage graphs.
\revision{The difference in the connectivity of the preimage graphs} of two faces $F_1$ and $F_2$, which are incident via a segment $f(ab)$, where $ab$ is an edge of $K$, is locally determined by the upper and lower link of $ab$.
\revision{
Assume that we cross from $F_1$ to $F_2$ via $f(ab)$ by crossing from the half-plane with negative orientation with respect to $f(a)$ and $f(b)$ to the half-plane with positive orientation. 
Then all triangles $abv$, where $v$ is in the lower link of $ab$, are removed from the preimage graph of $F_1$, and all triangles $abv'$, where $v'$ is in the upper link of $ab$, are added to the preimage graph of $F_1$ to form the preimage graph of $F_2$.
For the opposite direction of travel, all triangles $abv$ are added and all triangles $abv'$ are removed.}

Once all preimage graphs are computed, the \emph{correspondence graph} $H$ is constructed.
The vertices of the correspondence graph are the fiber components of the preimage graphs of all faces.
Edges connect vertices representing fiber components when there is a correspondence between them -- that is, when a fiber component passes from one face to another without \revision{any change in connectivity}.
Each sheet in the Reeb space is a collection of faces such that the connected components of the preimage graphs of those faces are in the same connected component in the correspondence graph.
The adjacency of the sheets in the Reeb space is determined by the adjacency of the faces that comprise them.

For an example of stages of the arrange and traverse algorithm refer to \cref{fig:arrange-and-traverse-example}.
In the bivariate case, the only geometric computations performed by this algorithm are in 2D -- computing the arrangement and orientation tests for upper and lower links.
The rest of the computation is entirely combinatorial.
This is a significant advantage over other Reeb space algorithms, as it allows the use of only 2D geometric predicates~\cite{hristovRobustGeometricPredicates2025}.
An exact open-source implementation using CGAL can be found on GitHub~\cite{hristov2025}.
Finally, this algorithm can be extended to higher dimensional domains $K^d$ and ranges $\real{k}$ such that $d > k$.

\begin{figure}[ht]
  \centering
  \subfloat[Arrangement traversal.]{
    \includegraphics[width=0.30\textwidth]{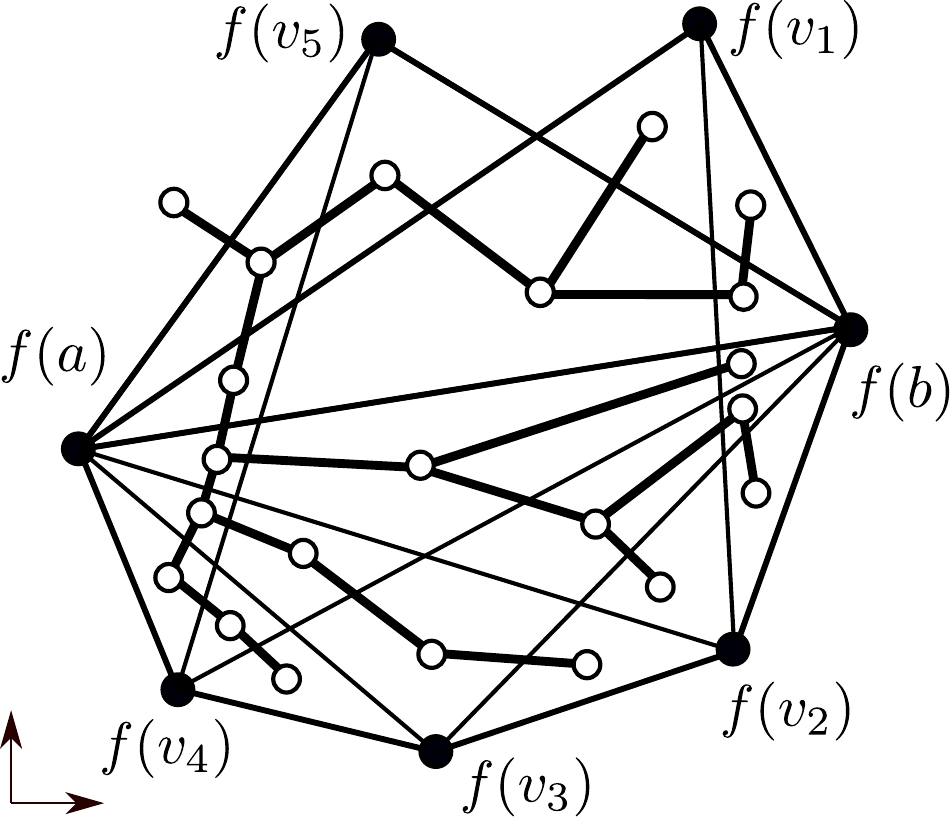}
  }
  \hfill
  \captionsetup{textformat=simple}
  \subfloat[Correspondence graph.]{
    \includegraphics[width=0.30\textwidth]{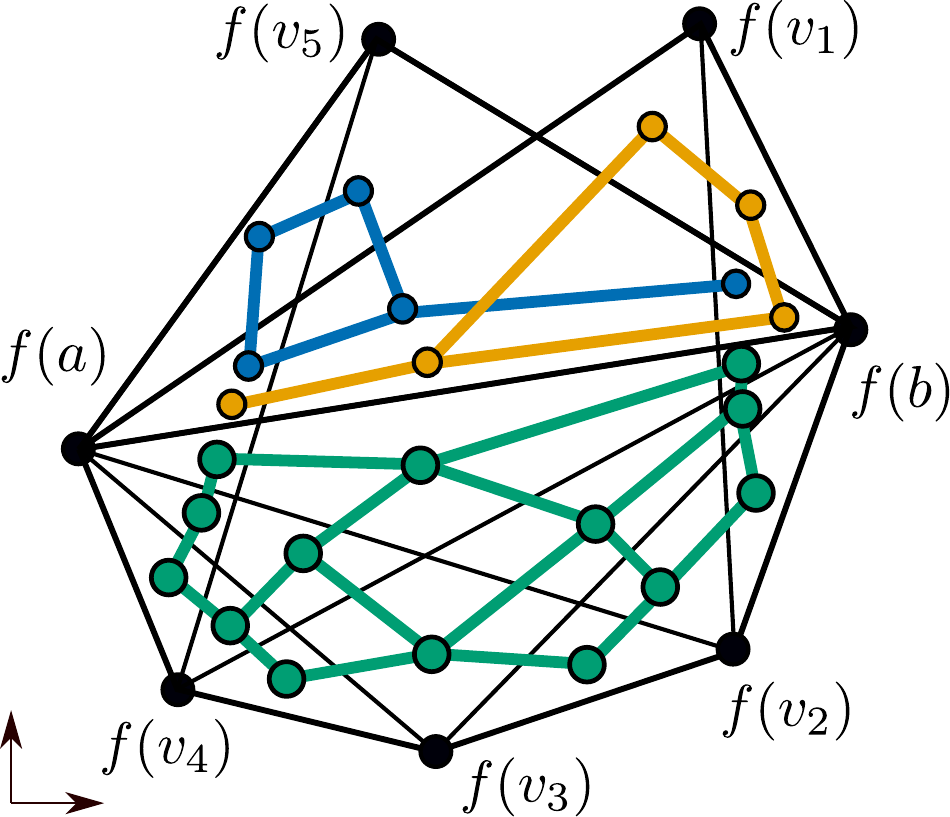}
    \label{fig:arrange-and-traverse-example-h}
  }
  \hfill
  \captionsetup{textformat=simple}
  \subfloat[Reeb space.]{
    \includegraphics[width=0.30\textwidth]{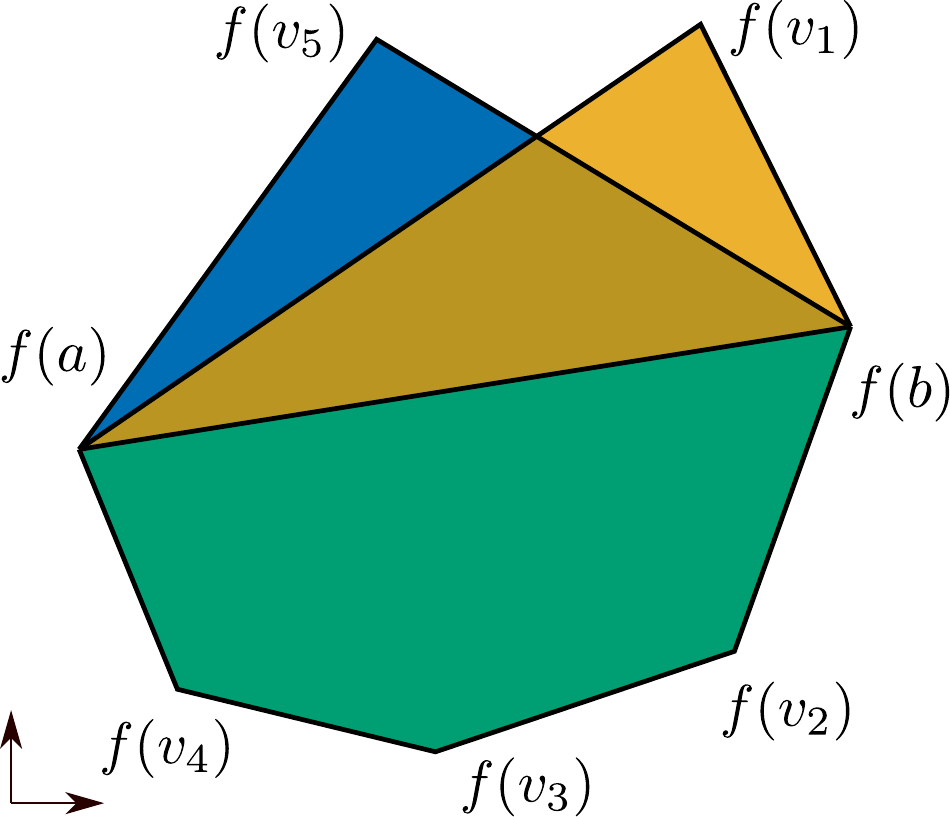}
  }
  \caption{The key stages of the arrange and traverse algorithm~\cite{hristovArrangeTraverseAlgorithm2025} for \cref{fig:domain-regular-arrangement}.}
  \label{fig:arrange-and-traverse-example}
\end{figure}

\subsection{Segment Intersections and Arrangements}
\label{sec:background-intersections}

\revision{
The intersections of $n$ line segments in the plane can be computed in time $O(n\log{n} + k)$, where $k$ is the number of intersections~\cite{chazelleOptimalAlgorithmIntersecting1992b, balaban1995optimal}.
A subdivision of the plane by segments into vertices, edges, and faces is called an \emph{arrangement} and can be computed in $O(n\log{n} + k\log{n})$ time~\cite{debergComputationalGeometryAlgorithms2008, goodmanHandbookDiscreteComputational2018} and represented with a \emph{doubly connected edge list (DCEL)}.
The special case of reporting the intersections between two sets of segments, typically labeled \emph{red} and \emph{blue}, can be solved in $O(n\log{n} + t)$ time, where $t$ is the number of intersections between the two sets.}

\section{Singular Arrange and Traverse Algorithm}
\label{sec:algorithm}

\revision{
The \emph{key idea} of the algorithm we propose is that \emph{not all faces and preimage graphs are needed} to compute the Reeb space in the arrange and traverse algorithm (see \cref{sec:background-arrange-and-traverse}). 
The only faces where the connectivity of fibers can change are the ones that are incident to at least one singular segment or singular vertex point.
We call these faces \emph{essential} because they outline a part of the boundary of at least one sheet of the Reeb space.
More precisely, they outline the images of sheets in the plane under $\pi$.
We refer to all other faces as \emph{nonessential} because they do not contribute to the shape of the sheets.
Our goal is to compute the Reeb space by \emph{using only the essential faces} and their preimage graphs.}








\subsection{Definitions}
Let $K$ be a triangulation of a $3$-manifold and let $f : |K| \to \real{2}$ be a generic \pl map. 
Let $S_f$ be the singular set of $f$ and let $R_f$ be the Reeb space of $f$ (see \cref{sec:preliminaries}).
We \emph{do not} assume that the singular set is simple.
For clarity of exposition, we will assume that $f(S_f)$ is connected -- we will deal with the case where it is not in \cref{sec:nested-faces}.
We will refer to preimage graphs as \emph{fiber graphs}, since fibers are defined as preimages of points under $f$.
Let $A = \textnormal{Arrangement}(f(K^{(1)}))$, where $K^{(1)}$ is the 1-skeleton of $K$, be the arrangement of all segments, and let $\bar{A} = \textnormal{Arrangement}(f(S_f))$ be the arrangement of all singular segments.
We will refer to $A$ and $\bar{A}$ as the \emph{full} and \emph{singular} arrangement respectively (see \cref{fig:singular-domain-arrangement}).

\begin{figure}[ht]
  \centering
  \subfloat[Singular edges (in bold).]{
    \includegraphics[width=0.30\textwidth]{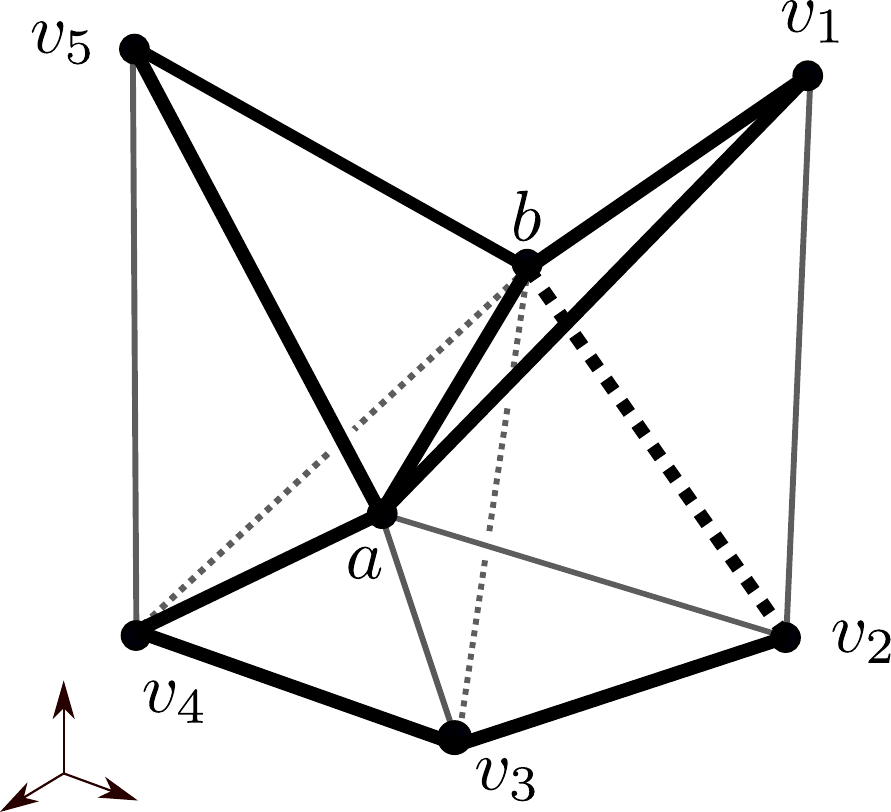}
  }
  \captionsetup{textformat=simple}
  \subfloat[Singular segments (in bold).]{
    \includegraphics[width=0.30\textwidth]{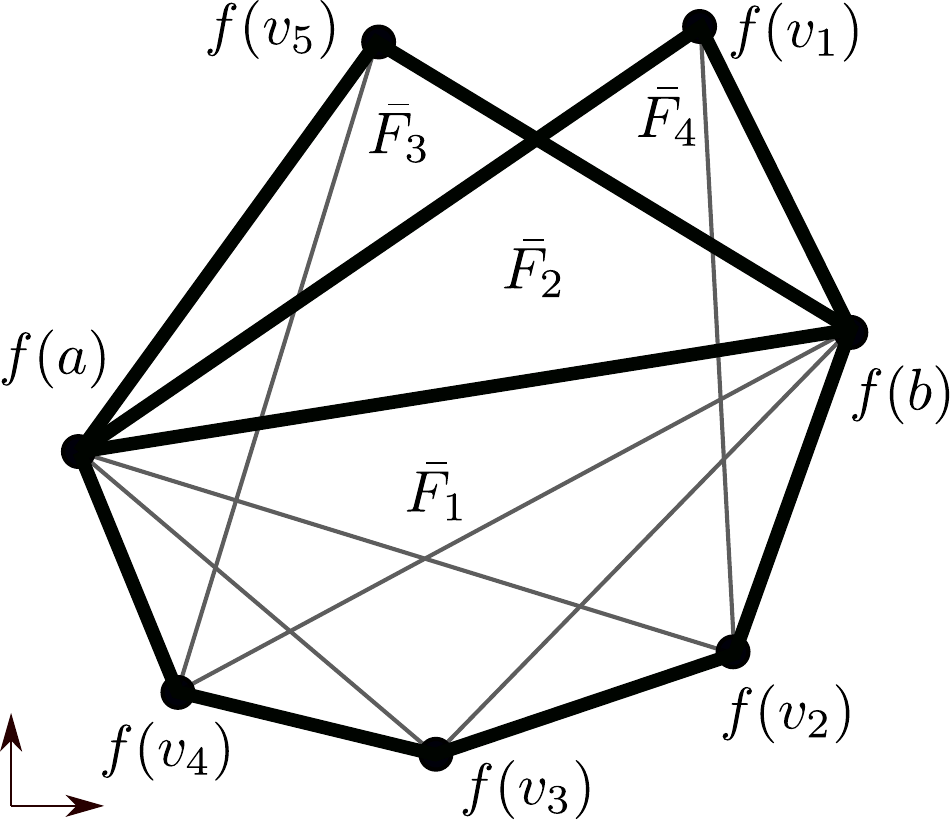}
  }
  \caption{
      The singular set (a) of \cref{fig:domain-regular-arrangement} and its image in the range (b).
      The arrangement $\bar{A}$ of the singular segments has four bounded faces, each subdivided by the faces of the full arrangement.
  }
  \label{fig:singular-domain-arrangement}
\end{figure}
Consider a face of the singular arrangement $\bar{F} \in \bar{A}$.
The fibers that correspond to points in the interior of $\bar{F}$ are homeomorphic to one another because they are all regular -- no change in connectivity happens to any of them.
Since $\bar{A}$ is an arrangement of a subset of the segments used in $A$, $\bar{F}$ is subdivided into a set of faces in $A$ such that $\bar{F} = \cup_{i}{F_i}$, where $F_i \in A$.
We describe the connectivity of fibers inside $\bar{F}$ by grouping the fiber graphs of all faces that subdivide $\bar{F}$ into a fiber graph class.

\begin{definition}[Fiber Graph Class]
    Let $\bar{F}$ be a face of the singular arrangement $\bar{A}$ and let $\{{F_{i}}\}_{i}$ be the set of faces of the full arrangement that are in $\bar{F}$.
    The fiber graph class $\bar{G}_{\bar{F}}$ of $\bar{F}$ is the set of fiber graphs $\bar{G}_{\bar{F}} = \{{G_{F_{i}}}\}_{i}$, where $G_{F_i}$ is the fiber graph of $F_{i}$.
\end{definition}

Fiber graph classes allow us to \emph{bundle} together fibers with the same topology to examine their connectivity in bulk.
For an example see \cref{fig:example-fiber-graph-class}.
In order to examine that connectivity, we define the notion of a component of a fiber graph class.

\begin{definition}[Fiber Graph Class Component]
    \revision{
    Let $\bar{G}_{\bar{F}} = \{{G_{F_{i}}}\}_{i}$ be the fiber graph class of a face $\bar{F}$ of $\bar{A}$.
    A component $\bar{C}$ of a fiber graph class $\bar{G}_{\bar{F}}$ is a set 
    $\{{C_i}\}_{i}$, where each $C_i$ is a connected component of $G_{F_i}$, such that any two connected components $C_i$ and $C_j$ of $\bar{C}$ correspond to each other (see correspondence in \cref{sec:background-arrange-and-traverse}).}
\end{definition}

The connected components of fiber graphs within a single fiber graph class component are in the same connected component of the correspondence graph $H$.
In order to determine how the components of fiber graph classes change their topology across the faces of the singular arrangement we define their correspondence.

\begin{definition}[Fiber Graph Class Component Correspondence]
\label{def:class-correspondence}
    \revision{
    Two components $\bar{C}$ and $\bar{C'}$ of two fiber graph classes $\bar{G}_{\bar{F}}$ and $\bar{G}_{\bar{F'}}$ correspond to each other if there exist two connected components $C \in \bar{C}$ and $C' \in \bar{C'}$ that correspond to each other.}
\end{definition}

\begin{figure}[ht]
  \centering
  \includegraphics[width=1.00\textwidth]{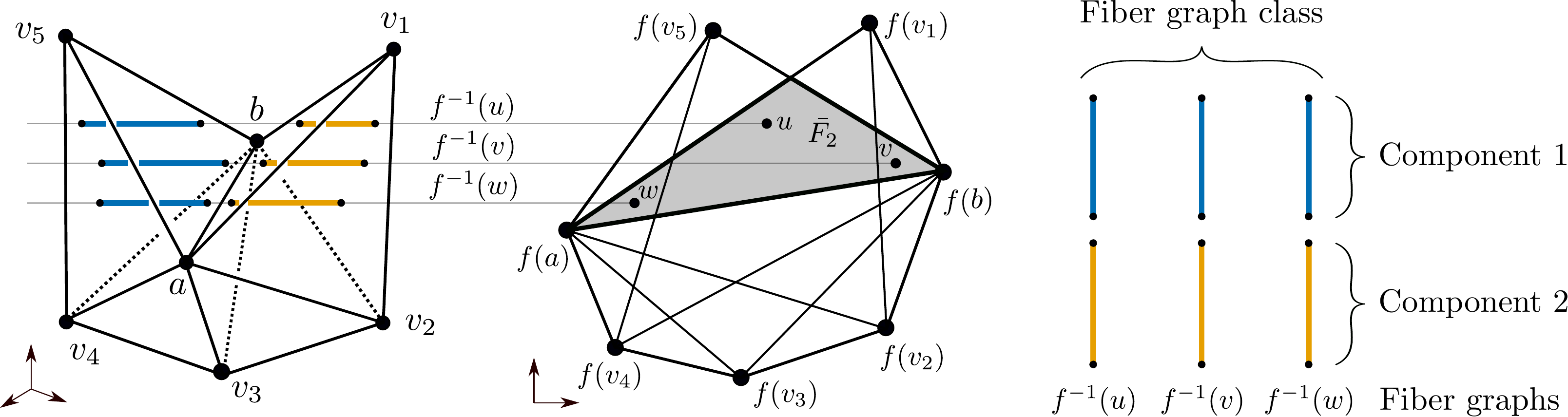}
  \caption{
      The fiber graph class of the face $\bar{F_2}$ is the collection of the fiber graphs of faces of the full arrangement that are in $\bar{F_2}$.
      All yellow (and respectively blue) fiber components correspond to one another; they change continuously from one to another as a fiber point moves around $\bar{F_2}$.
      All corresponding fiber graph components form a component of the fiber graph class of $\bar{F_2}$.
  }
  \label{fig:example-fiber-graph-class}
\end{figure}

We can describe the connectivity of fiber graph class components across all faces of the singular arrangement with the singular correspondence graph.

\begin{definition}[Singular Correspondence Graph]
    \revision{The singular correspondence graph $\bar{H}$ is a graph whose vertices are the set of all fiber graph class components and whose edges connect corresponding ones (see \cref{def:class-correspondence})}.
\end{definition}

The sheets of the Reeb space have a one-to-one matching with the connected components of the singular correspondence graph.
The geometry of a sheet of the Reeb space, or its image in the plane via $\pi$, is then the union of all faces $\cup_i\bar{F}_i$ whose fiber graph classes have components that span the connected component of $\bar{H}$ that matches that sheet.

By reframing the computation of the Reeb space in terms of fiber graph classes, their components, and the singular correspondence graph $\bar{H}$, we can optimize running time and memory efficiency in several ways.
First, we can completely skip the computation of all nonessential fiber graphs because the correspondence between two fiber graph class components can only be established with components of essential fiber graphs.
Second, we can save a significant amount of memory, which is the primary bottleneck of the arrange and traverse algorithm, by representing every fiber graph class component as a single vertex in $\bar{H}$.
Finally, since all essential faces are incident to singular segments, we do not have to compute the full arrangement $A$, only the singular one $\bar{A}$.
Then we can identify the essential faces by finding the intersections between $\bar{A}$ and the regular segments.

\subsection{Algorithm Description}
\label{sec:algorithm-desc}

We describe the \emph{singular arrange and traverse algorithm} in three stages.
First, we compute all essential faces in the \emph{geometric preprocessing stage}.
Then, we compute their fiber graphs and the singular correspondence graph in the \emph{topological traversal stage}.
Finally, we compute the geometry of the sheets and their adjacency in the \emph{geometric postprocessing stage}.

\paragraph*{Stage I. Geometric Preprocessing}
In the first stage, we compute all essential faces \emph{implicitly} without using the full arrangement $A$.
To do so, we compute the singular set $S_f$, then map it to the range to compute the singular arrangement $\bar{A}$, and finally we compute the intersection of $\bar{A}$ and the regular segments.
We store all regular segments that intersect the boundary $\partial \bar{F}$ of each face $\bar{F}$ of $\bar{A}$ in a \emph{circular counterclockwise ordered list} $Q_{\bar{F}}$.

When two regular segments intersect the interiors of two different half-edges or vertices of $\bar{F}$ we use the counterclockwise order defined on the boundary of $\bar{F}$ by the DCEL data structure of $\bar{A}$.
When two regular segments intersect the interior of the same half-edge of $\bar{F}$, we compare the barycentric coordinates of the points of intersection.
When two regular segments share a vertex of $\bar{F}$ as a common endpoint, we use their radial ordering. 
We use our running example to illustrate this in \cref{fig:example-singuar-red-blue}, where $Q_{\bar{F}_2} = \{f(v_1v_2), f(v_1v_2), f(v_4v_5), f(v_4v_5)\}$.

Each essential face $F \subseteq \bar{F}$ is implicitly represented with the part of its boundary that intersects $\partial \bar{F}$.
\revision{This can be identified by pairs of consecutive regular segments in $Q_{\bar{F}}$.}
Note also that an essential face may be represented multiple times when  $\partial F \cap \partial \bar{F}$ has multiple connected components.
Such is the case for the essential face whose boundary vertices include $i_1, i_2, i_3$ and $i_4$ in the middle of face $\bar{F_2}$ in \cref{fig:singular-domain-arrangement-blue-red}.

\begin{figure}[ht]
  \centering
  \subfloat[Singular arrangement.]{
    \includegraphics[width=0.30\textwidth]{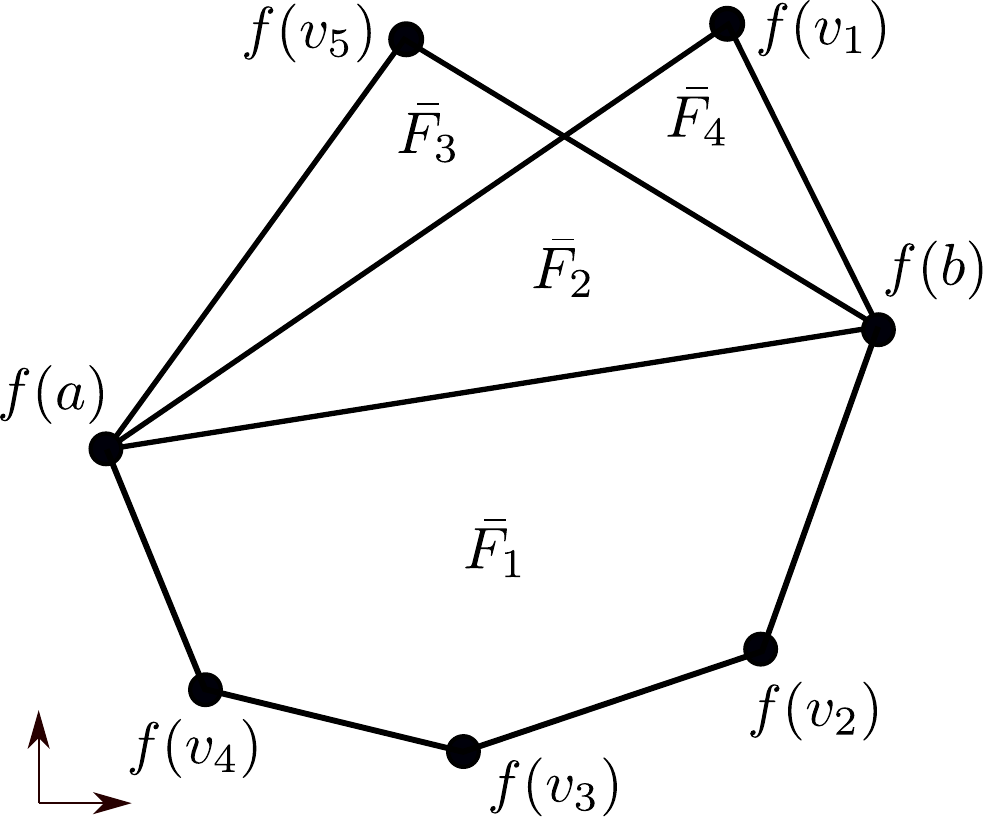}
  }
  \hfill
  \captionsetup{textformat=simple}
  \subfloat[Essential faces (shaded).]{
    \includegraphics[width=0.30\textwidth]{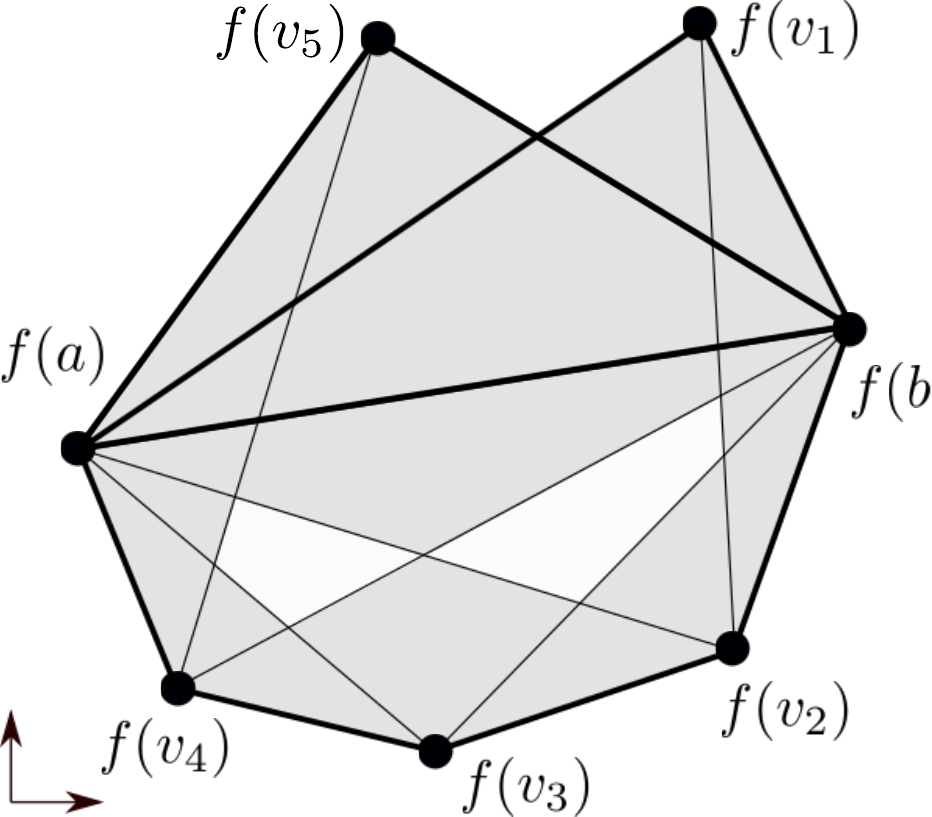}
  }
  \hfill
  \subfloat[Regular segment intersections.]{
    \includegraphics[width=0.30\textwidth]{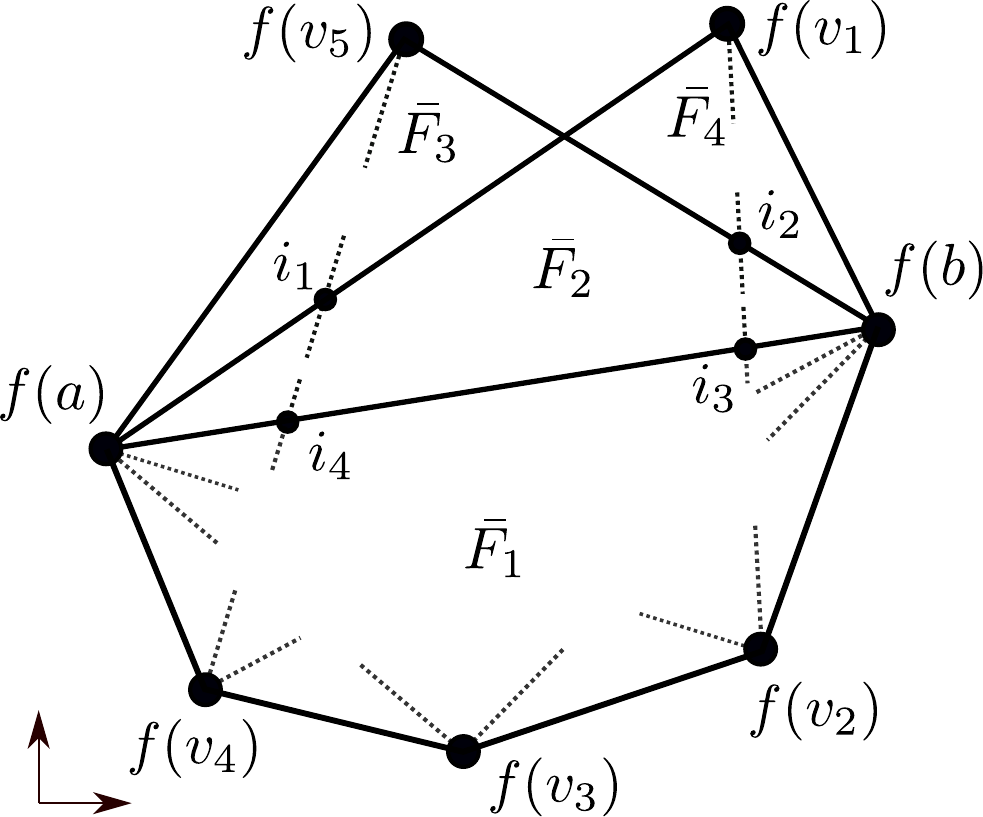}
    \label{fig:singular-domain-arrangement-blue-red}
  }
  \caption{
      Our algorithm computes the arrangement $\bar{A}$ of all singular segments (a), then idenfities all essential faces (b) by computing the intersections between $\bar{A}$ and all regular segments (c).
  }
  \label{fig:example-singuar-red-blue}
\end{figure}

\paragraph*{Stage II. Topological Traversal}
In the second stage, we traverse the faces of the singular arrangement along their shared segments \revision{using breadth-first search (BFS)} to compute the essential fiber graphs and the singular correspondence graph $\bar{H}$.
Starting with the unbounded face outside $\bar{A}$, whose fiber graph is empty, we visit one of its incident bounded faces $\bar{F}$ via a segment $f(ab)$, where $ab$ is an edge in $K$.
The edge $ab$ is definite \revision{(see \cref{sec:preliminaries})} because $f(ab)$ is on the boundary of $\bar{A}$, so we add the triangles $abv$, where $v$ is in the upper (or lower, depending on the direction of travel) link of $ab$, to the first essential fiber graph in the fiber graph class of $\bar{F}$.

Next, we compute all other essential fiber graphs of $\bar{F}$ with a loop operation.
A \emph{loop operation} implicitly loops around the inner boundary of $\bar{F}$ to visit all essential faces in $\bar{F}$ using the order $Q_{\bar{F}}$ we have established in the geometric preprocessing stage.
During the loop of $\bar{F}$, we cross from one essential face to the next via the regular segments listed in $Q_{\bar{F}}$.
\revision{We add and remove the appropriate triangles related to the upper and lower links of the edges that map to those segments (see \cref{sec:background-arrange-and-traverse})}.

\revision{
After we have computed the essential faces of $\bar{F}$, we visit its neighboring faces.
We compute the new fiber graphs by adding and removing triangles according to the upper and lower links of the singular edges that map to shared segments on the boundaries of incident singular faces.
We continue to loop each face as we visit it in the BFS traversal until we have computed the fiber graphs of all essential faces.
For an example of the loop operation, see \cref{fig:example-loop-one-face}, and for an example of the overall traversal, see \cref{fig:singular-arrange-and-traverse}.
}


\revision{
Finally, we construct the singular correspondence graph $\bar{H}$.
Let $\bar{F}$ and $\bar{F'}$ be two singular faces which are incident via the segment $f(ab)$, and let $F$ and $F'$ be any two essential faces in $\bar{F}$ and $\bar{F'}$ respectively, which are also incident via $f(ab)$.
We use the upper and lower links of $ab$ to determine which components in the fiber graphs of $F$ and $F'$ are not affected by the transition from $F$ to $F'$ via $f(ab)$ and establish correspondence between them~\cite{hristovArrangeTraverseAlgorithm2025}.
This allows us to establish correspondence between the components of the fiber graph classes of $\bar{F}$ and $\bar{F'}$.
We do this for all pairs of incident singular faces.
}

\begin{figure}[ht]
  \centering
  \subfloat[Essential fiber graphs of $\bar{F_1}$.\label{fig:example-loop-one-face-a}]{
    \includegraphics[width=0.33\textwidth]{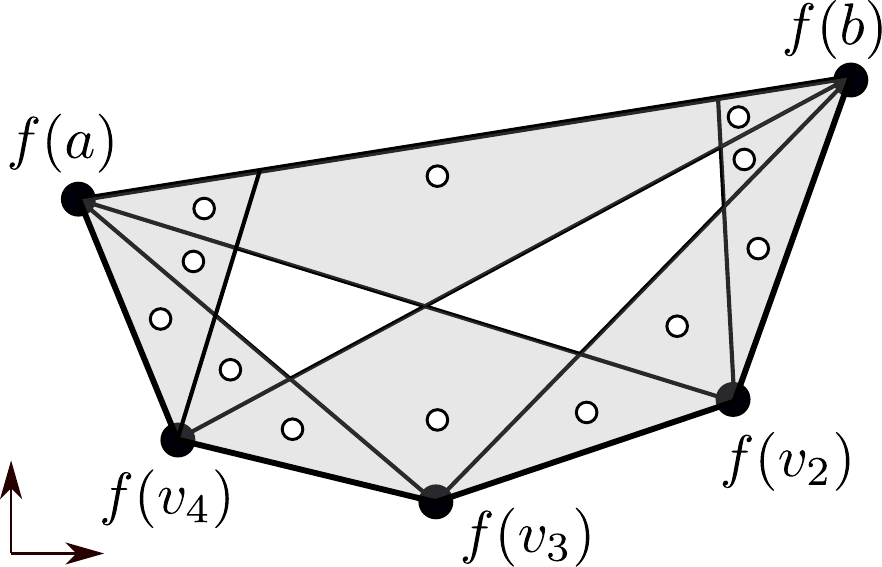}
  }
  \hspace{0.05\textwidth} 
  \captionsetup{textformat=simple}
  \subfloat[Looping all essential fiber graphs.\label{fig:example-loop-one-face-b}]{
    \includegraphics[width=0.33\textwidth]{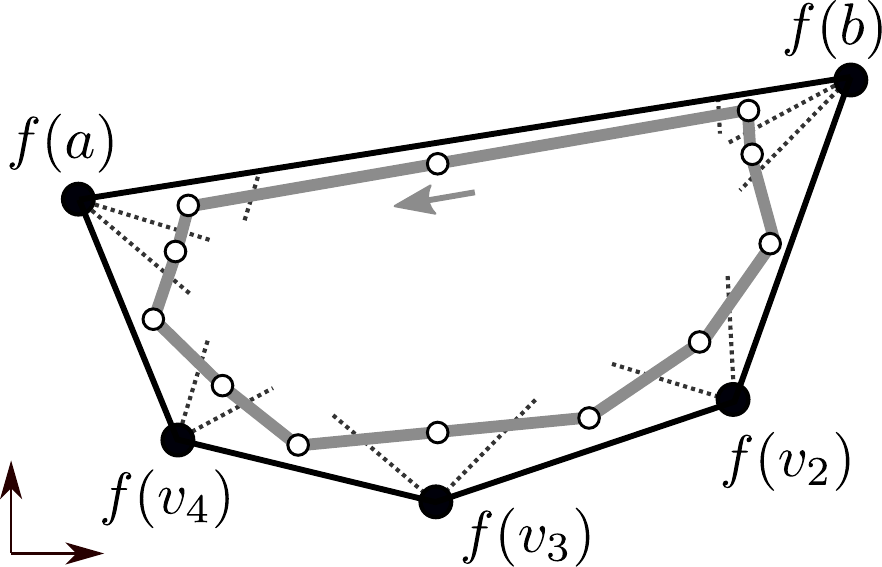}
  }
  \caption{
    Computing all essential fiber graphs of $\bar{F_1}$ from \cref{fig:singular-domain-arrangement} with a loop operation.
  }
  \label{fig:example-loop-one-face}
\end{figure}

\begin{figure}[ht]
  \centering
  \subfloat[Looping $\bar{F_1}$.]{
    \includegraphics[width=0.30\textwidth]{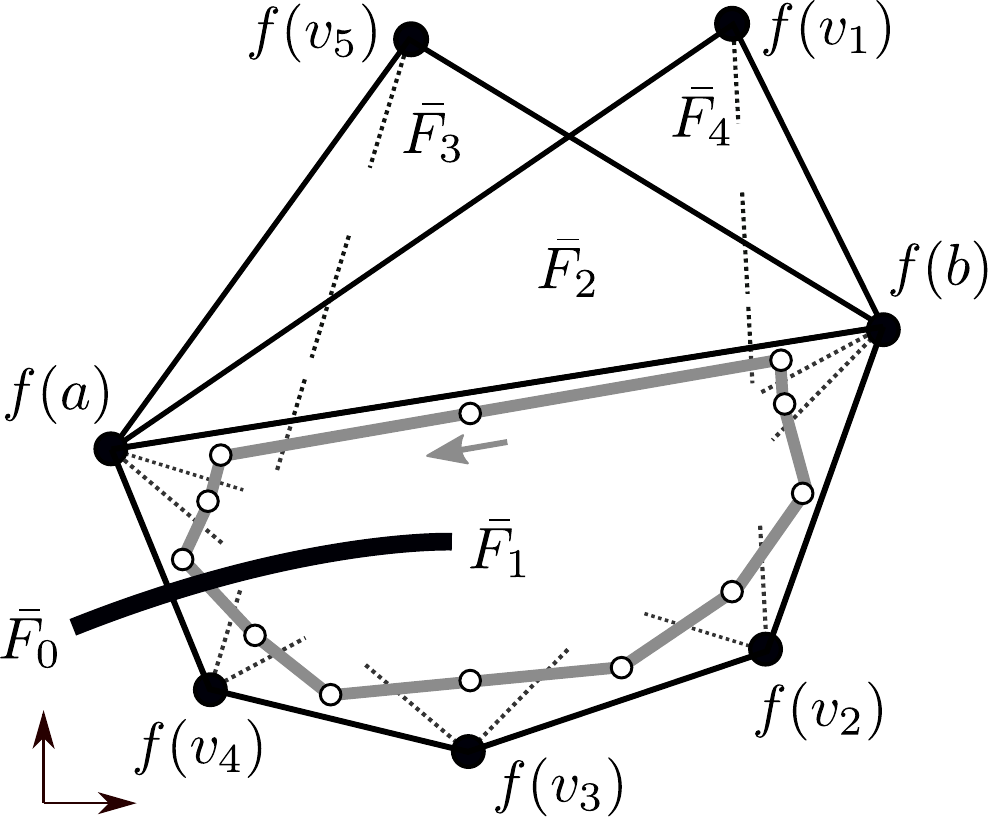}
  }
  \hfill
  \captionsetup{textformat=simple}
  \subfloat[Looping $\bar{F_2}$.]{
    \includegraphics[width=0.30\textwidth]{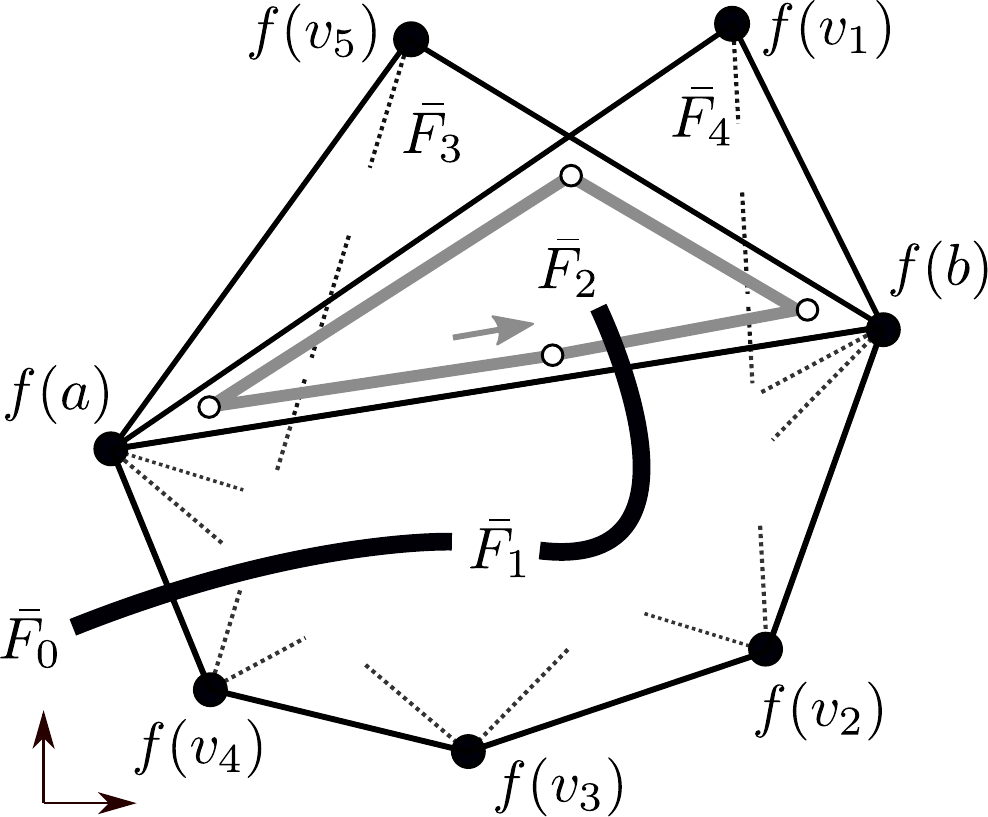}
  }
  \hfill
  \subfloat[Looping $\bar{F_3}$ \& $\bar{F_4}$.]{
    \includegraphics[width=0.30\textwidth]{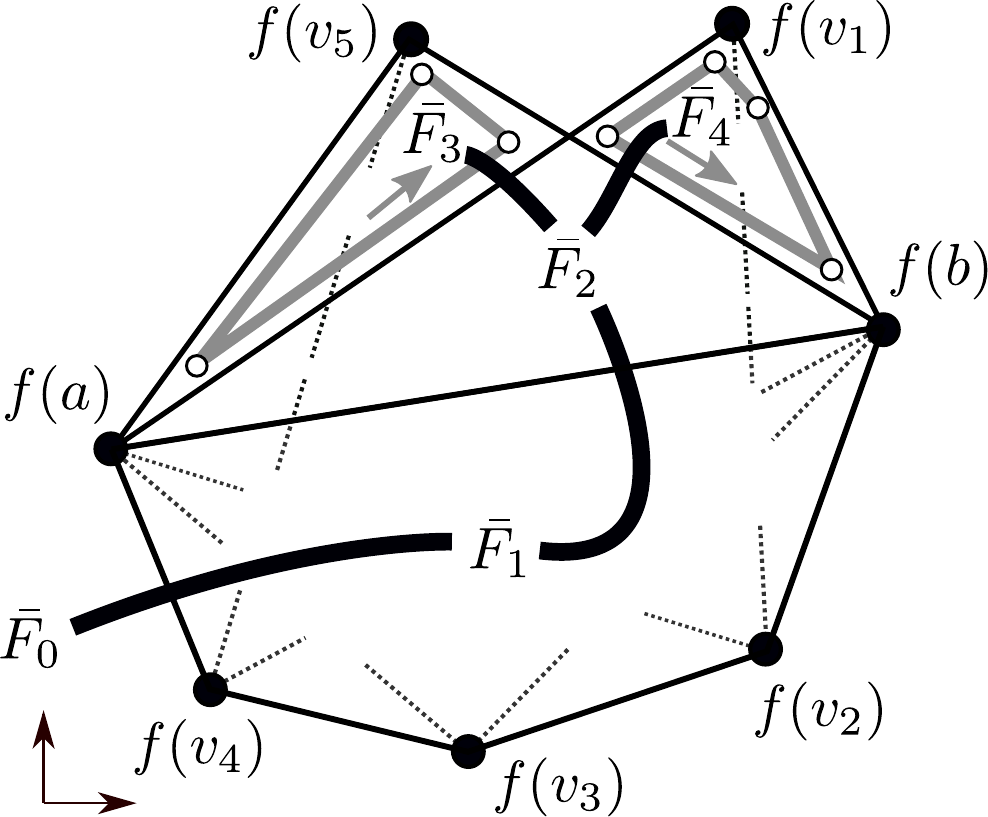}
  }
  \caption{Traversal of the singular arrangement and looping each face.}
  \label{fig:singular-arrange-and-traverse}
\end{figure}

\begin{figure}[ht]
  \centering
  \subfloat[Correspondence graph.]{
    \includegraphics[width=0.31\textwidth]{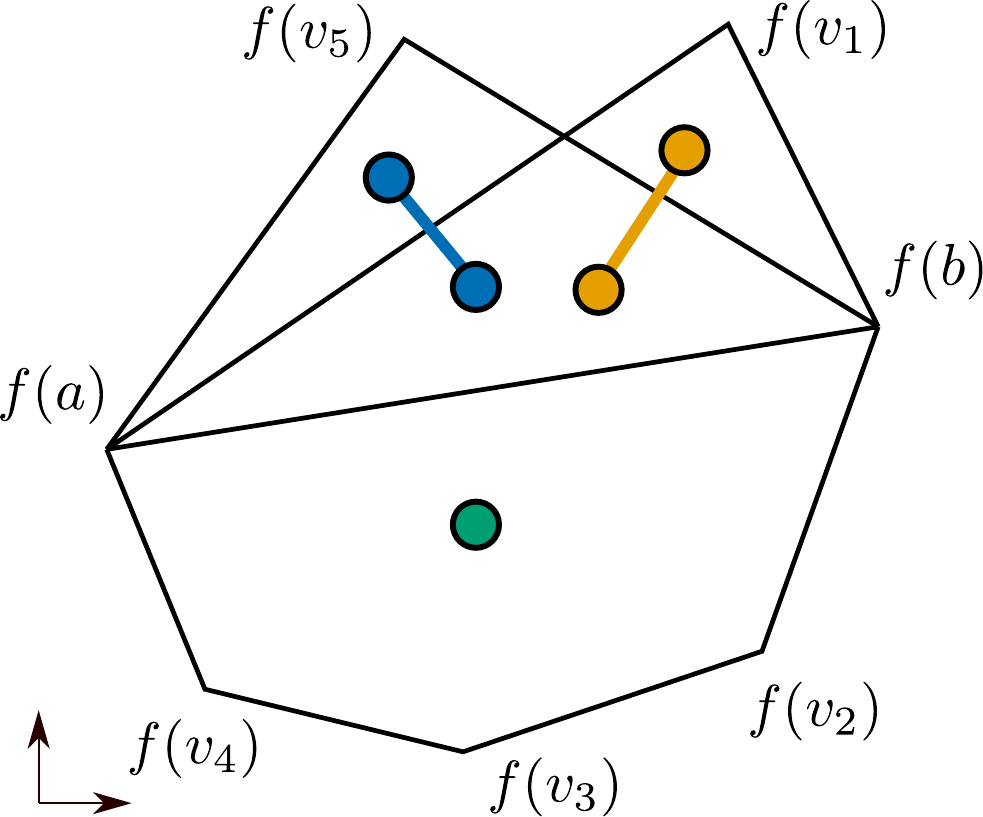}
  }
  \hspace{0.03\textwidth}
  \captionsetup{textformat=simple}
  \subfloat[Reeb space.]{
    \includegraphics[width=0.31\textwidth]{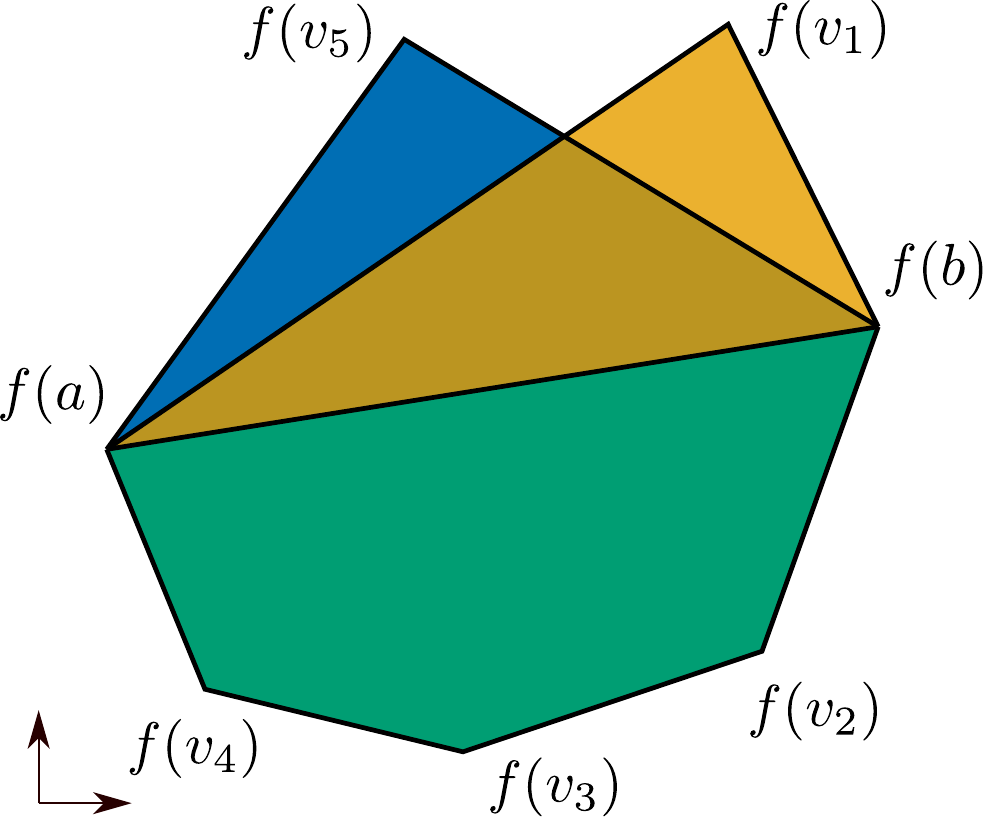}
  }
  \caption{
      Singular correspondence graph and Reeb space of the toy example from \cref{fig:domain-regular-arrangement}.
    }
  \label{fig:singular-correspondence-graph-and-reeb-space}
\end{figure}

\paragraph*{Stage III. Geometric Postprocessing}
Once we have computed the singular correspondence graph, we can compute the geometry and adjacency of the sheets of the Reeb space $R_f$.
By the geometry of a sheet of the Reeb space, we mean its image into the plane via $\pi: R_f \to \real{2}$.
Let $L$ be a sheet that matches a connected component of $\bar{H}$. 
Then, the geometry of $L$ is exactly the union of faces in the singular arrangement whose fiber graph class components span that component of $\bar{H}$.
Sheets are adjacent when they have incident faces in the singular arrangement. 
The singular correspondence graph and the Reeb space of our running example are given in \cref{fig:singular-correspondence-graph-and-reeb-space}.
Notice how much smaller the singular correspondence graph is compared to the correspondence graph in \cref{fig:arrange-and-traverse-example-h}.

\subsection{Nested Faces}
\label{sec:nested-faces}

We now address the case where the image of the singular set is \emph{not} connected.
In this case, either the unbounded face of the singular arrangement has more than one hole, or some bounded face has \revision{at least one hole}, or both.
If the unbounded face of the arrangement has multiple holes, we can traverse them independently.
If a bounded face $\bar{F_1}$ has a nested face $\bar{F_2}$ (a hole), we find the shortest path in the $1$-skeleton of $K$ that connects a vertex that maps to $\partial \bar{F_1}$ to a vertex that maps to $\partial \bar{F_2}$.
We mark the regular edges on this path as \emph{pseudo-singular} and compute the arrangement of all singular and pseudo-singular segments, which is now guaranteed to be connected.
After this, our algorithm proceeds as usual. 

Note that this may introduce additional faces in the singular arrangement, such as $\bar{F_3}$ in \cref{fig:nested-faces}, whose nonessential fiber graphs will be computed by our algorithm.
We will demonstrate that this is not an issue in practice and that only a negligible number of pseudo-singular edges are added for real-life datasets we evaluate in \cref{sec:evaluation}.

\begin{figure}[ht]
  \centering
    \includegraphics[width=0.26\textwidth]{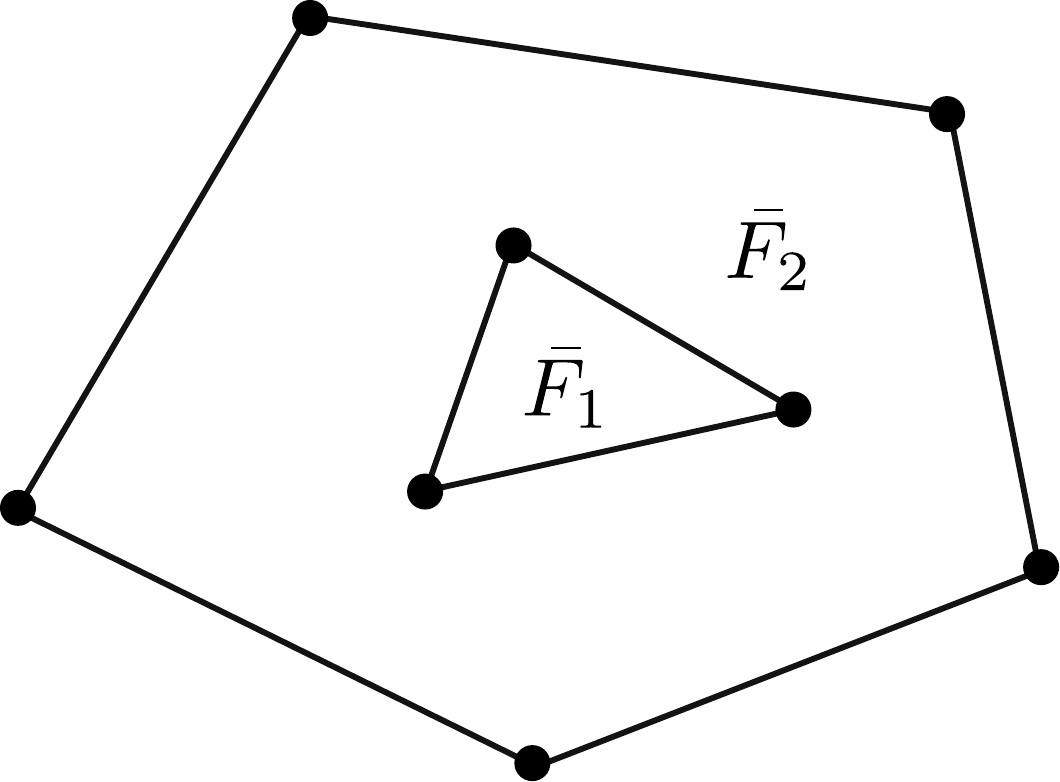}
  \hspace{0.1\textwidth}
    \includegraphics[width=0.26\textwidth]{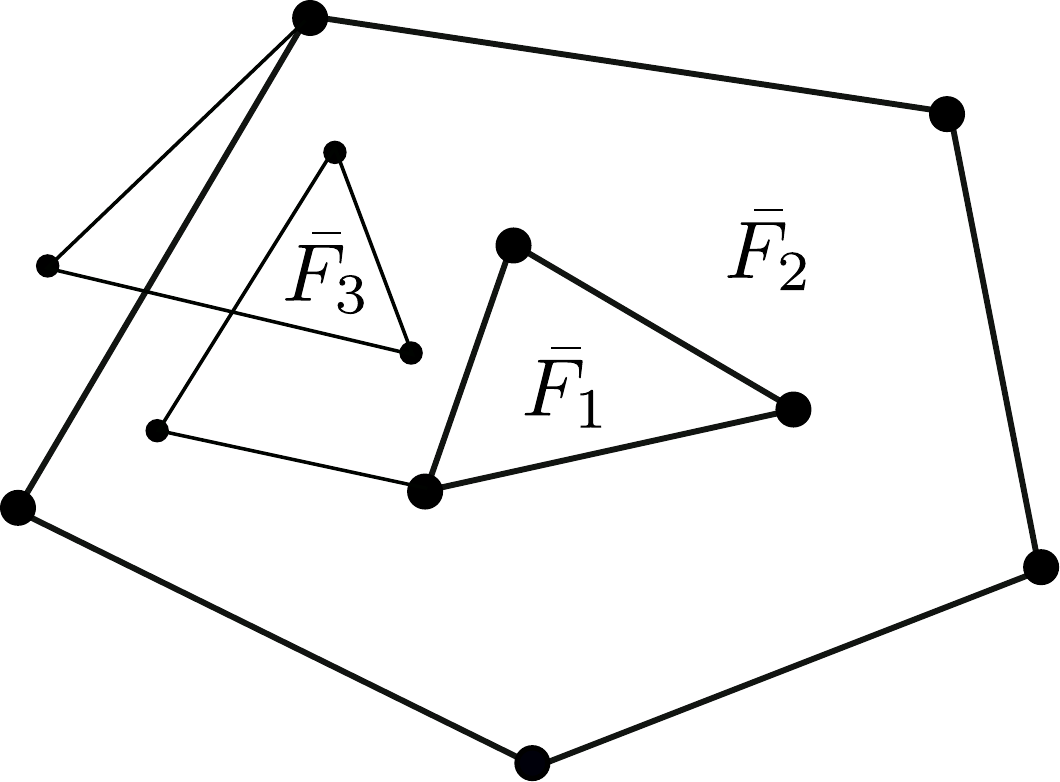}
  \caption{
      Connecting nested faces by adding pseudo-singular segments.
  }
  \label{fig:nested-faces}
\end{figure}

\subsection{Higher Dimensional Domains}
\revision{Finally, we describe an extension of our algorithm for computing the Reeb space of a \pl map $f : |K| \to \real{2}$, where $K$ is a triangulation of a $d$-manifold with $d > 3$.}
Due to the skeleton lemma~\cite{edelsbrunnerReebSpacesPiecewise2008b}, the Reeb space of $f : |K| \to \real{2}$ is homeomorphic to the Reeb space of $f' : K^{(3)} \to \real{2}$, which is the restriction of $f$ to the $3$-skeleton of $K$.
In order to compute the Reeb space of $f'$, we compute all singular edges of $K$~\cite{edelsbrunnerJacobiSetsMultiple2002} and discard all simplices of dimension bigger than three from $K$.
Then, we map the singular edges to the plane, compute their arrangement, and traverse it to compute the singular correspondence graph.
The only difference is that we need a new subroutine to track the connectivity of fiber graphs.
This is because the regular fibers of $f'$ are not \pl curves but the $1$-skeleton of triangulated $(d-2)$-manifolds.
This requires more involved data structures for dynamic graph connectivity tracking, similar to how Reeb graph algorithms handle the scalar case~\cite{doraiswamyEfficientAlgorithmsComputing2009, parsa2012}.

\section{Proofs}
\label{sec:proofs}

In this section we prove the correctness and worst-case running time of our algorithm using our definitions from the start of \cref{sec:algorithm}.

\subsection{Correctness}
\label{sec:correctness}

First, we prove the correctness of each stage in our algorithm (see \cref{sec:algorithm-desc}).

\begin{lemma}[Essential Faces]
    The singular arrange and traverse algorithm correctly identifies all essential faces.
    \label{lemma:essential-faces}
\end{lemma}

\begin{proof}
    Let $F$ be an essential face of the full arrangement $A$, and let $\bar{F}$ be the face of the singular arrangement $\bar{A}$ that contains $F$.
    The case when $F = \bar{F}$ is trivial because then $F$ is already correctly represented in $\bar{A}$.
    Suppose that $F \subset \bar{F}$.
    Let $B$ be a connected component of $\partial F \cap \partial \bar{F}$.
    Then $B$ is a path from $i_1$ to $i_2$, where $i_1$ and $i_2$ are the intersection points of two regular segments $s_1$ and $s_2$ with $\partial \bar{F}$.

    First, we examine the case where $B = i_1 = i_2$.
    In the DCEL of the full arrangement $A$, the half-edges of the boundary cycle of $F$ that originate from $s_1$ and $s_2$ must be consecutive.
    This means that no other regular segment comes between $s_1$ and $s_2$ in the counterclockwise order of segments around $i_1$ in $A$. 
    Otherwise, that other segment would be on the boundary of $F$.
    Therefore, the segments $s_1$ and $s_2$ are present as a consecutive pair in the circular list $Q_{\bar{F}}$ of regular segments of $\bar{F}$.
    If on the other hand $i_1 \neq i_2$, then $B$ is subdivided into at least one half-edge in the DCEL of $A$.
    No other regular segment intersects the interior of $B$, because then the essential face $F$ would have been subdivided differently in $A$.
    Similarly to the previous case, other regular segments whose endpoint is either $i_1$ or $i_2$ are either before $s_1$ or after $s_2$ in the counterclockwise order of segments around those vertex points.
    Therefore, $s_1$ and $s_2$ are consecutive in $Q_{\bar{F}}$.

    Conversely, we can use a similar argument to show that each consecutive pair of regular segments in $Q_{\bar{F}}$ can be uniquely identified with an essential face.
\end{proof}

\begin{lemma}[Essential Fiber graphs]
    The singular arrange and traverse algorithm correctly computes all essential fiber graphs.
    \label{lemma:essential-fiber-graphs}
\end{lemma}
\begin{proof}
    Since all essential faces are identified correctly by \cref{lemma:essential-faces}, our algorithm can traverse them, analogously to the arrange and traverse algorithm, but in a different order.
    To see why this does not affect the correctness, consider a path $P$ in the range that only intersects regular or singular segments, from the unbounded outside face to an essential face.
    Any such path gives rise to a sequence of additions and removals of triangles, based on the upper and lower links of edges whose segments $P$ intersects. 
    This sequence can be used to build the fiber graph of the essential face, which is the same as the fiber graph computed by our algorithm.
    Additional pseudo-singular segments in $\bar{A}$ (see \cref{sec:nested-faces}) do not affect the correctness because our algorithm still identifies such a path $P$ correctly.
\end{proof}

\begin{lemma}[Singular Correspondence Graph]
    The connected components of the singular correspondence graph correctly represent the sheets of the Reeb space.
    \label{lemma:singular-correspondence-graph}
\end{lemma}
\begin{proof}

    The singular correspondence graph $\bar{H}$ can be obtained from the correspondence graph $H$ by contracting all vertices that are associated with the components of fiber graphs within the same fiber graph class.
    A contraction of a graph does not affect its connectivity. 
    Each vertex of $\bar{H}$ simply represents a larger area in the range of fiber points whose fibers \revision{have the same topology.
    The singular correspondence graph can be computed correctly by establishing correspondence only between the essential fiber graphs because nonessential fiber graphs do not introduce correspondences with new fiber graph classes. 
    They only correspond directly to other components within their own class, and it is precisely these correspondences that are lost in the contraction of $H$ to $H'$.}
    The correctness of the correspondence procedure between fiber graphs is proven in the original arrange and traverse algorithm~\cite{hristovArrangeTraverseAlgorithm2025}.
\end{proof}

\begin{corollary}[Reeb space]
    The singular arrange and traverse algorithm correctly computes the Reeb space of a generic \pl map $f: |K| \to \real{2}$, where $K$ is a triangulation of a 3-manifold.  
\end{corollary}
\begin{proof}
    By \cref{lemma:essential-faces} and \cref{lemma:essential-fiber-graphs}, the essential faces and essential fiber graphs are computed correctly.
    By \cref{lemma:singular-correspondence-graph}, the singular correspondence graph represents the sheets of the Reeb space and is computed correctly.
    Therefore, the singular arrange and traverse algorithm correctly computes the Reeb space of $f$.
\end{proof}

\subsection{Complexity}
\label{sec:complexity}

Let $N_e$ and $N_t$ be the number of edges and triangles in $K$ respectively.
Let $N_s$ be the number of singular edges.
Let $k_s$ be the number of intersections between all singular segments, and let $k_r$ be the number of intersections between all regular and singular segments.
First, we compute the arrangement of all singular segments in time $O((N_s + k_s)\log{N_s})$.
Next, we compute the number of intersections between all regular and singular segments and their relative order in time $O((N_e + k_r)\log{N_e})$.
We use a red-blue segment intersection algorithm (see \cref{sec:background-intersections}) and sort the intersection points along each singular segment using their barycentric coordinates.
This concludes the geometric part of our computation.
What remains is the complexity of computing the essential fiber graphs and the singular correspondence graph, which we demonstrate with the following lemma.

\begin{lemma}
    Computing the essential fiber graphs and the singular correspondence graph has worst-case running time $O(N_sN_t \log N_t)$.
\end{lemma}
\begin{proof}
    Let $k(e)$ be the number of times an edge $e \in K$ is visited by the algorithm.
    A visit can mean either crossing a singular segment in the traversal from one face to another in the singular arrangement, or crossing a regular segment in the looping of a single face.
    \revision{When we visit an edge with degree $\textnormal{deg}(e)$, where we define the degree of an edge as the number of its incident triangles, we add and remove at most $\textnormal{deg}(e)$ triangles from the relevant fiber graphs, and
    perform at most $\textnormal{deg}(e)$ connectivity queries.}
    Each of those operations takes at most $O(\log N_t)$ time if we represent each fiber graph with a balanced search tree~\cite{colemclaughlinLoopsReebGraphs2003}.

    The overall worst-case running time can then be obtained by summing the costs of all visits across all edges $O\big(\sum_{e \in E} k(e) \textnormal{deg}(e) \log(N_t)\big)$, where $E$ is the set of edges in $K$.
    Since we do not consider any intersections between regular segments, each segment is intersected, and hence visited, at most $N_s$ times. 
    Using $k(e) \leq N_s$ we obtain $O\big(N_s\log(N_t)\sum_{e \in E}\textnormal{deg}(e)\big)$.
    Using a straightforward counting argument, we can show that $\sum_{e \in E}\textnormal{deg}(e) = 3N_t$. 
    This leads to a worst-case running time of $O(N_sN_t \log N_t)$.
\end{proof}

Computing the connected components of $\bar{H}$ can be done in linear time in its size with a graph search.
The geometry of the sheets of the Reeb space can be obtained by combining all faces of the singular arrangement that belong to each sheet.
Thus, \emph{the overall complexity of our algorithm is $O(N_sN_t \log N_t)$}.
\revision{Note that, practical performance depends primarily on $k_s$ and $k_r$, which determine how many times an edge is actually visited.}

When $O(N_s) = O(N_e)$ our algorithm has the same running time as arrange and traverse.
However, our algorithm is significantly more efficient for practical data sets where the number of singular edges is orders of magnitude less than the number of regular edges, as we will demonstrate in our evaluation in \cref{sec:evaluation}.

\section{Implementation}
\label{sec:implementation}

We provide an open-source implementation of the singular arrange and traverse algorithm in C++.
The input is a tetrahedral mesh with two scalar values defined at each vertex, represented in the \texttt{.vtu} file format with VTK~\cite{vtkBook}. 
For the geometric preprocessing operations we used CGAL 6.0.2~\cite{cgal:eb-24b} with the exact predicates and exact constructions kernel~\cite{cgal:bfghhkps-lgk23-24b}.
We compute the red-blue intersections with CGAL's \texttt{box\_intersection\_d}~\cite{cgal:kmz-isiobd-24b}, which has quadratic complexity.
\revision{
While, this is not optimal, it has an efficient implementation.}

We track changes in the connectivity of fiber graphs with a graph search over the affected~\cite{hristovArrangeTraverseAlgorithm2025} fiber components.
While this operation is not logarithmic in the size of the fiber, in practice most of the connectivity computations are done within the loop operation for a face of the singular arrangement, so all triangles can be directly added or removed from the same fiber graph component without affecting the connectivity.

Another important practical optimization, similar to arrange and traverse, is to only keep the fiber graphs at the wave front of the BFS search in memory since those are the only ones needed to establish correspondence to create the singular correspondence graph~\cite{hristovArrangeTraverseAlgorithm2025}.
While degenerate cases, \revision{which occur when the input \pl map is not generic (see \cref{sec:preliminaries})}, can be handled robustly with simulation of simplicity~\cite{hristovRobustGeometricPredicates2025}, random numerical perturbation with double floating-point precision was sufficient for our tests.

\section{Evaluation}
\label{sec:evaluation}

In this section, we perform a practical evaluation of the efficiency of our algorithm in terms of running time and memory.
We benchmark our algorithm against the arrange and traverse algorithm~\cite{hristovArrangeTraverseAlgorithm2025}, since that is the only other open-source exact Reeb space implementation known to the authors.
Our tests were performed on a PC with two Intel Xeon Gold 6130 processors with 32 cores clocked at 2.1GHz and 384GB memory with Rocky Linux 9.6.
We measured execution time using the C++ \texttt{chrono} library and estimated memory usage via the process's resident set size by reading \texttt{/proc/self/status}.

We benchmark our algorithm on three numerical simulation datasets -- MVK~\cite{ chakrabortyDecipheringMethylationEffects2023a, sharma2025}, ethane-diol~\cite{carrFiberSurfacesGeneralizing2015, 8017627} and enzo~\cite{bryanENZOADAPTIVEMESH2014, carrFiberSurfacesGeneralizing2015}.
MVK comes from a quantum chemistry simulation studying photo-induced changes in electron transition orbitals. 
Ethane-diol is a molecular dynamics simulation of a molecule's electron density field and its reduced gradient. 
Lastly, the origin of Enzo is cosmological simulation of the universe's expansion; here, the bivariate field represents matter and dark matter concentration. 
Each dataset is given on a three-dimensional regular grid and then tetrahedralized using Paraview's~\cite{ahrens2005paraview} default filter.

\revision{
We summarize our results in \cref{tab:performance}. 
Columns labeled $\times$ show the speedup of our algorithm over arrange and traverse. 
In the time columns, this indicates how many times faster it is; in the memory columns, how much less memory it uses.
The timings for the geometric preprocessing stage are given in the \textit{Prep.} column and the timings for the topological traversal stage are given in the $R_S$ column.
The timings for the geometric postprocessing stage are omitted, as the computation time is negligible compared to the first two stages and depends on how the Reeb space will be used in downstream tasks.}
Each entry is rounded to the nearest whole number.
We downsampled each dataset multiple times to obtain a wide range of data sizes and proportion of the number singular edges, ranging from 0.20\% for MVK to 3.73\% for ethane-diol and finally 43.78\% for enzo.
The enzo datasets and ethane-diol with 3,150K tetrahedra required the use of perturbation with strength $0.0001$.

\begin{table*}[h]
\centering
\footnotesize
\begin{tabular}{p{0.03\linewidth}
                R{0.06\linewidth} 
                R{0.05\linewidth} 
                R{0.07\linewidth}R{0.03\linewidth}R{0.08\linewidth}R{0.05\linewidth} 
                R{0.09\linewidth}R{0.05\linewidth}R{0.08\linewidth}R{0.05\linewidth}} 
\toprule
\textbf{\rotatebox{90}{}} 
  & \multicolumn{1}{c}{\boldmath$N_T$} 
  & \multicolumn{1}{c}{\boldmath$N_s / N_e$} 
  & \multicolumn{4}{c}{\textbf{Time}} 
  & \multicolumn{4}{c}{\textbf{Memory}} \\
\cmidrule(lr){4-7} \cmidrule(lr){8-11}  
  & & 
  & Prep.(s) & $\times$ & $R_S$(s) & $\times$ 
  & Prep.(Mb) & $\times$ & $R_S$(Mb) & $\times$ \\
\midrule
\multirow{4}{*}{\rotatebox{90}{MVK}}
  & 22K    & 2.93\%   & 2    & 14 & <1   & 2K  & <1    & 5K  & <1     & 1K \\
  & 53K    & 1.70\%   & 3    & 28 & <1   & 6K  & <2    & 12K & <1     & 4K \\
  & 189K   & 0.78\%   & 8    & 72 & 2   & 25K & <2    & 64K & <1     & 18K \\
  & 1,540K & 0.20\%   & 41   & ---  & 6   & --- & <2    & --- & <1     & --- \\
\midrule
\multirow{4}{*}{\rotatebox{90}{\shortstack{Ethane-\\diol}}}
  & 25K    & 12.09\% & 85     & 0.9  & 30   & 129   & 204   & 58   & 237   & 23 \\
  & 46K    & 9.33\%  & 153    & 1.0  & 50   & 204   & 311   & 84   & 433   & 26 \\
  & 113K   & 6.50\%  & 402    & 1.4  & 123  & 426   & 600   & 139  & 810   & 46 \\
  & 384K   & 3.73\%  & 1,843  &  ---   & 556  & ---     & 1,610 & ---    & 2,059 & ---  \\
  & 3,150K & 1.44\%  & 19,259 &  ---   & 5,238 & ---     & 7,438 & ---    & 7,389 & --- \\
\midrule
\multirow{4}{*}{\rotatebox{90}{Enzo}}
  & 1K   & 49.88\% & 8     & 0.2  & 15     & 8     & 107   & 3    & 124    & 6 \\
  & 4K   & 48.53\% & 189   & 0.1  & 323    & 13    & 869   & 4    & 988    & 17 \\
  & 9K   & 44.62\% & 2,535  & 0.1  & 4,012  & 23    & 3,543 & 4    & 4,743  & 30 \\
  & 14K  & 43.78\% & 7,127  & ---    & 11,294 & ---     & 7,882 & ---    & 11,140 & --- \\
\bottomrule
\end{tabular}
\caption{
Performance metrics for real-life datasets. 
We report the number of input tetrahedra $N_T$, the percentage of singular edges $N_s$ relative to all edges $N_e$, computation times, and memory usage for both preprocessing and Reeb space computations, including scaling factors ($\times$) relative to arrange and traverse.
Runs with over a day of compute were terminated and marked with  "---".
}
\label{tab:performance}
\end{table*}

Our \emph{primary result} is that the singular arrange and traverse algorithm is up to four orders of magnitude faster for MVK at 189K tetrahedra. 
For this dataset, arrange and traverse requires over 13 hours of computation time and about 100 GB of memory.
Singular arrange and traverse not only takes less than 3 seconds and uses less than 700 MB of memory, but is also able to compute the Reeb space at full resolution (1,540K tetrahedra) in less than 10 seconds and 6 GB of memory.
This speedup is due to the small number of singular edges, which our algorithm can successfully exploit.

Although the two other datasets have a larger singular set, we still see a significant speedup -- at least two orders of magnitude for ethane-diol and at least one for enzo.
Note that even though the number of singular edges in enzo is about half the total number of edges, we still see a significantly higher speedup -- not just two-fold but up to 23x.
This is due to the fact that the singular arrange and traverse algorithm not only allows us to avoid computing nonessential faces, but also allows us to group the computation of regular fiber graphs and significantly reduce the size of the correspondence graph.

In the case of MVK and ethane-diol, we are able to compute the Reeb space of their original non-downsampled resolution.
This is something that would have been extremely difficult with arrange and traverse.
The enzo dataset, however, has 83 million tetrahedra in its non-downsampled resolution and a high proportion of singular edges, which makes it difficult for our algorithm to handle.

We verified the correctness of our algorithm by comparing the sheets it outputs against arrange and traverse.
For all datasets for which we could compute arrange and traverse, the outputs of both algorithms were identical.
Finally, we note that the amount of additional computation needed to connect the nested faces was negligible. In the MVK datasets our algorithm added one to two additional pseudosingular edges to connect the arrangement and avoid nested faces.
In the ethane-diol datasets that number was from 20 to 94.

The geometric preprocessing stage of our algorithm is significantly faster for MVK and ethane-diol and uses substantially less memory.
Our memory usage improvement holds for enzo as well, but it takes up to 10 times longer to compute, even though arrange and traverse computes the arrangement of all segments.
This \revision{may} be addressed by implementing a more efficient algorithm for red-blue segment intersection.

\section{Conclusion}
\label{sec:conclusion}

In this paper, we presented the singular arrange and traverse algorithm for the exact and efficient computation of Reeb spaces of bivariate \pl maps.
This algorithm takes advantage of the fact that datasets with a small number of singular edges, relative to the total number of edges, have smaller Reeb spaces, which allows us to avoid nonessential computation from previous Reeb space algorithms such as arrange and traverse~\cite{hristovArrangeTraverseAlgorithm2025}.

We presented a reformulation of the computation in terms of fiber graph classes and the singular correspondence graph, enabling the collective treatment of many related fiber graphs and thereby improving efficiency. We further provided an exact implementation of this approach and benchmarked it against the original arrange and traverse algorithm, demonstrating speedups ranging from one to four orders of magnitude across multiple real-world datasets derived from numerical simulations.

Future work will focus on optimising the running time through parallelisation as well as simplification for datasets with large singular sets.
Other open problems include non-heuristic methods for dealing with nested faces and extensions to higher dimensional ranges.

\bibliography{singular-at}

\end{document}